\documentclass[11pt,english]{article}
\usepackage[T1]{fontenc}
\usepackage[utf8]{inputenc}
\usepackage{xcolor}
\usepackage{amstext}
\usepackage{amsthm}
\usepackage[authoryear]{natbib}
\PassOptionsToPackage{normalem}{ulem}
\usepackage{ulem}

\makeatletter

\providecolor{lyxadded}{rgb}{0,0,1}
\providecolor{lyxdeleted}{rgb}{1,0,0}
\DeclareRobustCommand{\mklyxadded}[1]{\textcolor{lyxadded}\bgroup#1\egroup}
\DeclareRobustCommand{\mklyxdeleted}[1]{\textcolor{lyxdeleted}\bgroup\mklyxsout{#1}\egroup}
\DeclareRobustCommand{\mklyxsout}[1]{\ifx\\#1\else\sout{#1}\fi}

\theoremstyle{plain}
\newtheorem{thm}{\protect\theoremname}
\theoremstyle{definition}
\newtheorem{defn}[thm]{\protect\definitionname}

\usepackage[margin=1in]{geometry}
\usepackage{amsmath,amssymb,amsthm}
\usepackage{booktabs,tabularx}
\usepackage{natbib}
\usepackage[hidelinks]{hyperref}
\usepackage{setspace}
\usepackage{enumitem}
\usepackage{tikz}
\usepackage{xcolor}
\usetikzlibrary{positioning}

\newtheorem{proposition}{Proposition}
\newtheorem{lemma}{Lemma}

\theoremstyle{remark}

\newcommand{\dd}{\,\mathrm{d}}

\makeatother

\usepackage{babel}
\providecommand{\definitionname}{Definition}
\providecommand{\theoremname}{Theorem}

\begin{document}
\title{Tradeable Import Certificates for Strategic Supply Security\thanks{This paper generalizes and supersedes an earlier working paper circulated
under the title “Tradeable Import Certificates: A Promising Instrument
to Support Domestic Production in Strategic Sectors?” \citep{Kranz2025TIC}.}}
\author{Sebastian Kranz, Ulm University}
\date{September 2026}
\maketitle
\begin{abstract}
Recent crises have made supply security central to trade policy. We
show how tradeable import certificates (TIC) implement targets for
domestic production and reliable foreign supply while preserving gains
from trade. A single certificate market per country decentralizes
the welfare-maximizing allocation under heterogeneous targets, with
certificate prices adjusting endogenously as conditions change. TIC
robustly protect these targets against a range of deviations from
trade agreements. A model of economic coercion microfounds the targets,
linking them to shortage salience, supply reliability, bilateral trade
dependence, and norms against yielding to coercion. Tariff-subsidy
agreements require more information and are more vulnerable to hidden
deviations. While a common carbon price provides a natural focal point
for climate agreements, quantity-based security targets could be a
more natural focal point for trade agreements than tariff and subsidy
rates.

\end{abstract}
{\small\medskip{}
}\\
{\small Keywords: }{\small\textit{tradeable import certificates, economic
security, strategic supply targets, trade agreements, economic coercion,
friendshoring}}{\small\par}

\section{Introduction}

The European Union aims to have domestic manufacturing capacity meet
at least 40 percent of its annual deployment needs for net-zero technologies
by 2030 \citep{EUNZIA2024}. China set a goal of meeting 70 percent
of its semiconductor demand from domestic production \citep{CRS2020Semiconductors}.
The U.S. Department of Defense wants a domestic rare-earth supply
chain able to cover all defense needs \citep{USDODRareEarths2024}.
After a pandemic that exposed the fragility of medical and semiconductor
supply chains, a war that turned energy dependence into a strategic
liability, and geopolitical rivalry in which market access itself
is used as a weapon, quantitative targets like these have become the
language in which governments state their economic-security goals.

How to meet such targets is far less clear. Governments combine production
subsidies, tax credits, preferential procurement, local-content rules,
tariffs, and export controls, with effects that are hard to predict
\citep{RotunnoRutaVerma2026}. Spillovers to trading partners and
their responses create a risk that uncoordinated economic-security
policies inefficiently fragment the international trading system \citep{ClaytonMaggioriSchreger2024}.
This tension is increasingly explicit in WTO reform discussions: how
can governments pursue legitimate industrial and security objectives
without eroding the gains from an open, rules-based trading system
\citep{WTOReport2026,EUWTOIndustrialPolicy2026}? An international
agreement could in principle coordinate these policies, but an agreement
written on tariff and subsidy rates faces two problems. The rates
that implement a given set of targets depend on technologies, demand,
transport costs, and other conditions that negotiators do not fully
observe and that change over time. And whatever rates are agreed,
a partner can undermine another country's target through hidden subsidies
or non-tariff barriers, which are notoriously hard to police \citep{HornMaggiStaiger2010,Gulotty2022}.

This paper proposes to write trade agreements directly around strategic
targets and to implement them through a robust market-based mechanism.
Each country combines its security objectives across products and
trading relationships into a single strategic target and implements
it with a matched system of tradeable import certificates (TIC). TIC
generalize Warren Buffett's \citeyearpar{Buffett2003} idea of import
certificates to balance the U.S. trade account. In his proposal exports
create certificates that exporters can sell and imports require certificates
that importers must surrender. Economically, the certificate price
thus works like a tariff for importers and like an export subsidy
for exporters.

TIC retain the basic logic but apply it to linear strategic supply
targets. Each production or trade flow earns or requires certificates
according to its contribution to the target, with coefficients that
may vary by product, origin, and destination. Strategically valuable
exports, secure domestic production, and reliable imports can create
certificates, while less reliable imports and domestic production
that depends on fragile upstream inputs can require them.

We establish three main results. First, the competitive equilibrium
under an agreement in which all countries implement their strategic
targets with matched TIC is constrained efficient. It maximizes ordinary
world welfare subject to every country's target, even as underlying
economic conditions change. The system is self-financing: no public
funds are required. Moreover, certificate prices make the marginal
welfare effects of the parameters of a country's strategic target
observable.

Second, a country's matched TIC robustly protects its strategic target
against a broad class of foreign policy deviations, including subsidies,
trade barriers, and quantity restrictions. This is a direct accounting
implication of the certificate mechanism. From an economic perspective,
automatic adjustment of the certificate price neutralizes the effect
of such policy measures on a binding strategic target.

Third, a bargaining model of economic coercion provides a microfoundation
for a class of \textit{linear surplus targets}. A rival threatens
a trade crisis to extract a concession, and the concession it can
extract depends on each side's politically salient shortages, the
reliability of supply from third countries, bilateral trade dependence,
and norms against yielding to coercion. We characterize for each country
the target that prevents coercion in a minimally restrictive way.
When every country adopts its corresponding target and implements
it with a matched TIC scheme, the resulting profile maximizes ordinary
welfare among all strategic target profiles that prevent coercion.

Our approach complements a growing literature that derives optimal
resilience and geoeconomic policy from disruption risk, bargaining
in the shadow of conflict, and national-security externalities \citep{GrossmanHelpmanLhuillier2023,ClaytonMaggioriSchreger2026,ClaytonMaggioriSchreger2024,Kooi2025,BeckoOConnor2025,Thoenig2023},
surveyed by \citet{MohrTrebesch2025}, and older work on protection
under embargo threats \citep{Mayer1977,Thompson1979,AradHillman1979,BergstromLouryPersson1985}.
That literature asks what security objective a country should pursue;
we ask how heterogeneous objectives, once chosen, can be implemented
in an agreement that preserves as much trade as they allow. Treating
non-economic objectives as constraints has a classical precedent \citep{Johnson1965,Bhagwati1967,BhagwatiSrinivasan1969}.
Institutionally, TIC are related to domestic-content, export-performance,
and import-export linkage requirements \citep{Grossman1981,HeranderThomas1986,Rodrik1987},
and to tradeable permits and quota licenses \citep{Dales1968,Montgomery1972,Anderson1987,KrishnaTan1999};
what is distinctive is the use of endogenously supplied, target-weighted
certificates to implement an agreement centered directly on supply
security targets. Our robustness result connects to work on incomplete
trade agreements and substitution toward non-negotiable instruments
\citep{Copeland1990,BagwellStaiger2001,Ederington2001,HornMaggiStaiger2010,LimaoTovar2011,BeverelliBoffaKeck2019}.
In climate agreements, a common carbon price provides a natural focal
point \citep{Weitzman2014,SchmidtOckenfels2021}; for trade agreements,
quantity-based security targets could instead be a more natural focal
point than tariff and subsidy rates. The same target-centered perspective
motivates our discussion of WTO rules.

Section \ref{sec:twocountry} illustrates the core mechanism and main
results in a simple two-country model. It compares an agreement based
on TIC with a calibrated tariff-subsidy agreement and with the Nash
equilibrium under uncoordinated trade policy to highlight advantages
and remaining limitations of the certificate approach. Section \ref{sec:general_tic}
develops the general model and establishes the main robustness and
efficiency results. It also studies TIC alongside other policy instruments,
including trade-related environmental measures such as the EU's CBAM,
and shows that their protective effects do not stack on top of one
another. Section \ref{sec:microfoundation} develops the microfoundation.
Section \ref{sec:discussion} discusses WTO rules and reform proposals,
related certificate systems, price-versus-quantity focal points, and
market power in certificate markets. We close with a discussion of
how TIC systems may become particularly relevant in a world with transformative
AI. Proofs are in the appendix.

\section{Motivating example\label{sec:twocountry}}

\subsection{Setting}

This section motivates our core ideas with a simple example. There
is a set $\mathcal{N}=\{A,B\}$ of two countries, indexed by $i$
and $j$, and a continuum of strategic products $m\in[0,1]$. Demand
is completely inelastic and each country always consumes one unit
of every product $m$, yielding a fixed total demand across all products
of 
\begin{equation}
Y_{i}=1\label{eq:xY}
\end{equation}
in each country $i$. Let $q_{ij}(m)\geq0$ denote the quantity of
product $m$ that is produced by country $i$ and sold in country
$j$. An allocation $q$ consists of a complete specification of all
$q_{ij}(m)$ satisfying $q_{ij}(m)+q_{jj}(m)=1$ for all $m$ and
all $i\ne j$. Let 
\begin{equation}
Q_{i}=\int_{0}^{1}\sum_{j\in\mathcal{N}}q_{ij}(m)\dd m\label{eq:xQ}
\end{equation}
denote total production by country $i$. Country $i$ produces product
$m$ at constant marginal cost $c_{i}(m)$. Worldwide total costs
for an allocation $q$ are given by
\begin{equation}
C(q)=\int_{0}^{1}\sum_{i\in\mathcal{N}}\sum_{j\in\mathcal{N}}c_{i}(m)q_{ij}(m)\dd m.\label{eq:xT}
\end{equation}
Let $V$ denote the fixed gross benefit of consumers and
\begin{equation}
W(q)=V-C(q)\label{eq:xW}
\end{equation}
denote the world ordinary welfare, which does not account for strategic
concerns. Since demand is inelastic, maximizing $W(q)$ is the same
as minimizing world production cost $C(q)$.

We assume that product markets are perfectly competitive and that
there are no transportation costs. The production cost difference
satisfies
\begin{equation}
c_{A}(m)-c_{B}(m)=-m_{0}+m,\qquad m_{0}\in(0,1).\label{eq:twogap-1}
\end{equation}
This is the continuum-of-goods linear Ricardian structure of \citet{DornbuschFischerSamuelson1977}.
Under free trade, country $A$ produces all products $m<m_{0}$ for
domestic and export markets and country $B$ produces all products
$m>m_{0}$.\footnote{How production is split over the measure-zero set of products where
countries have equal costs (in free trade, the product $m=m_{0}$)
is irrelevant for our analysis.} Country $A$'s total free-trade production is $Q_{A}^{FT}=2m_{0}$.
We assume $m_{0}<\frac{1}{2}$ so that country $A$ is a net importer
under free trade.

Country $A$ has the strategic target 
\begin{equation}
Q_{A}\geq\sigma_{A}Y_{A}=\sigma_{A}\qquad\text{with }2m_{0}<\sigma_{A}<1.\label{eq:x_targetA}
\end{equation}
This means domestic production must at least cover a share $\sigma_{A}$
of consumption. The target lies above the free trade level but still
allows positive net imports in the strategic sector. Given fixed domestic
demand $Y_{A}=1$, the target is equivalent to an absolute production
target $Q_{A}\geq\sigma_{A}$. Country $B$ has the strategic target
to produce at least its domestic demand $Q_{B}\geq Y_{B}=1$.

\subsection{Constrained efficiency}

A constrained efficient allocation $q^{*}$ maximizes ordinary welfare
$W(q)$ given the strategic targets of all countries as constraints.
Here
\begin{align}
q^{*}\in\arg\max_{q}\quad & W(q)\label{eq:x_ce_problem}\\
\text{s.t.}\quad & Q_{A}\geq\sigma_{A}\label{eq:x_A_constr}\\
 & Q_{B}\geq1\label{eq:x_B_constr}
\end{align}

The constrained efficient allocation has a simple structure: country
$A$ produces all products $m<m^{*}\equiv\frac{1}{2}\sigma_{A}$.
Figure 1(a) provides an illustration. $A$'s strategic constraint
(\ref{eq:x_A_constr}) is binding, i.e. $Q_{A}=\sigma_{A}$, while
country $B$'s remains slack. 100 percent of the products $m$ are
still traded.

\begin{figure}

\centering
\begin{tikzpicture}[x=5.8cm,y=3.0cm,font=\small]

\def\mfree{0.24}   
\def\qexp{0.23}    
\def\qhalf{0.48}   
\def\qdom{0.73}    
\def\slope{0.95}
\def\ymin{-0.58}
\def\ymax{0.78}

\begin{scope}
  \draw[->] (0,0) -- (1.04,0) node[right] {$m$};
  \draw[->] (0,\ymin) -- (0,\ymax); 
  \node[anchor=south east] at (1.1,\ymax) {$c_A(m)-c_B(m)$};

  \draw[gray] (0,0) -- (1,0);

  \draw[thick] (0,{-\slope*\mfree}) -- (1,{\slope*(1-\mfree)});

  \draw[dashed] (\qhalf,\ymin+0.02) -- (\qhalf,\ymax-0.02);
  \draw (\qhalf,0.025) -- (\qhalf,-0.025);
  \draw (\mfree,0.025) -- (\mfree,-0.025);

  \node[anchor=north] at (\mfree,-0.08) {$m_0$};
  \node[anchor=north] at (\qhalf,-0.08) {$\sigma_A/2$};

  \node[align=center] at (0.20,0.30) {$A$ exports};
  \node[align=center] at (0.78,0.30) {$B$ exports};

  \node[below=10pt] at (0.5,\ymin) {(a) Constrained efficient allocation};
\end{scope}

\begin{scope}[xshift=7.3cm]
  \fill[gray!20] (\qexp,\ymin+0.04) rectangle (\qdom,\ymax-0.08);
  
  \draw[->] (0,0) -- (1.04,0) node[right] {$m$};
  \draw[->] (0,\ymin) -- (0,\ymax);
  \node[anchor=south east] at (1.1,\ymax) {$c_A(m)-c_B(m)$};
 
  \draw[gray] (0,0) -- (1,0);

  \draw[thick] (0,{-\slope*\mfree}) -- (1,{\slope*(1-\mfree)});

  \draw[dashed] (\qexp,\ymin+0.02) -- (\qexp,\ymax-0.02);
  \draw (\qhalf,0.025) -- (\qhalf,-0.025);

  \draw[dashed] (\qdom,\ymin+0.02) -- (\qdom,\ymax-0.02);


  \node[anchor=north] at (\qhalf,-0.08) {$\sigma_A/2$};

  \node[align=center] at (0.11,0.30) {$A$\\exports};
  \node[align=center] at (0.87,0.30) {$B$ \\ exports};
  \node[align=center] at (0.48,0.50) {no\\trade};

  \node[below=10pt] at (0.5,\ymin) {(b) Nash equilibrium outcome};
\end{scope}
\end{tikzpicture}

\caption{Constrained efficient and Nash equilibrium allocations in our example.}
\label{fig:nashintuition}
\end{figure}
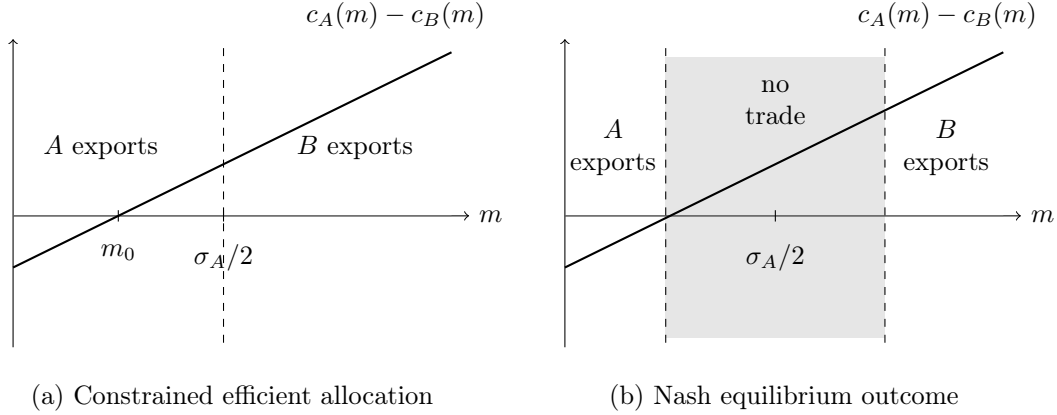

\subsection{Uncoordinated tariffs and subsidies}

We now extend the example to a game-theoretic model to illustrate
how uncoordinated trade policy using tariffs and subsidies can produce
an outcome in which country $A$ also exactly achieves its strategic
target while the production allocation is not constrained efficient.

Each country $i$ chooses an import tariff $t_{i}\geq0$ and export
subsidy $e_{i}\geq0$ that are applied uniformly over all products
$m$. Resulting costs for product $m$ exported from country $i$
to country $j$ are thus $c_{i}(m)+t_{j}-e_{i}$. Let $D_{i}$ denote
country $i$'s direct economic cost consisting of consumer expenditures
plus expenditures for export subsidies minus tariff income.

Country $A$ has the utility function
\begin{align}
u_{A} & =-\gamma_{A}\max\{\sigma_{A}-Q_{A},0\}-D_{A}\label{eq:uApreview}
\end{align}
where the constant $\gamma_{A}>0$ measures country $A$'s relative
weight on the strategic target compared to its direct economic costs.
The kink at the target resembles the loss-averse political objective
of \citet{FreundOzden2008}: shortfalls below the target are penalized,
overshooting brings no gain.

Country $B$ has the utility function 
\begin{align}
u_{B} & =\gamma_{B}Q_{B}-D_{B},\label{eq:uBpreview-1}
\end{align}
with a general strategic preference for more production measured by
$\gamma_{B}>0$.

In practice, the relationship between true preferences over worldwide
production allocations and officially pursued strategic targets is
likely to be complex. For example, openly pursuing targets that stipulate
export surpluses in strategic sectors likely carries a strong diplomatic
cost. Thus, even if country $B$ would, for large $\gamma_{B}$, prefer
to be the sole worldwide producer of strategically relevant products,
we do not assume that this preference translates into a strategic
target of $Q_{B}\geq2$.

We characterize the Nash equilibrium for a case where both countries
have strong strategic concerns, with a particularly high value of
$\gamma_{A}$.

\begin{proposition}[]\label{prop:nashpreview} Assume $\max\{2m_{0},1/2\}<\sigma_{A}<1$
and $\gamma_{A}-\sigma_{A}>\gamma_{B}>\frac{3}{4}-\frac{1}{2}\sigma_{A}$.
Suppose that each country simultaneously chooses an import tariff
$t_{i}\geq0$ and an export subsidy $e_{i}\geq0$. A Nash equilibrium
exists, and every Nash equilibrium satisfies:
\begin{enumerate}
\item Country $A$ exactly achieves its target, $Q_{A}=\sigma_{A}$, but
the allocation is not constrained efficient.
\item Both countries choose tariffs strictly larger than export subsidies:
$t_{A}>e_{A}\geq0$ and $t_{B}>e_{B}>0$.
\item At most 50 percent of products $m$ are internationally traded.
\end{enumerate}
\end{proposition}

Figure 1(b) illustrates the production allocation in a Nash equilibrium.
There is an inefficient no-trade zone of half the products. The reason
for that no-trade zone is that countries set tariffs higher than export
subsidies for fiscal reasons. Export subsidies benefit foreign consumers
while their costs are borne by domestic taxpayers.\footnote{A similar effect arises for general production subsidies which also
partly benefit foreign consumers whenever there are exports.} In contrast, when neglecting their distortion of the production allocation,
tariffs are cost neutral because higher domestic consumer prices are
exactly offset by tariff revenue. A country supplements tariffs with
export subsidies only when the production inefficiency from further
import substitution becomes large relative to the fiscal cost of boosting
exports.

\subsection{Trade agreement using TIC}

Consider now a trade agreement in which each country operates a TIC
scheme and uses no further tariffs or subsidies. Every imported unit
into country $i$ from the covered sector requires one tradeable import
certificate. Every exported unit from country $i$ creates exactly
one certificate that can be freely sold by exporters. Furthermore,
country $i$ auctions $a_{i}\geq0$ additional certificates. The resulting
revenues reduce direct costs $D_{i}$. The TIC scheme implies that
\begin{equation}
Q_{i}^{imp}\leq Q_{i}^{exp}+a_{i}\label{eq:preview_tic_constraint}
\end{equation}
where $Q_{i}^{imp}$ and $Q_{i}^{exp}$ denote the imports and exports
of country $i$ respectively. A simple reformulation yields a minimum
production guarantee that follows from accounting alone, independent
of any additional trade instruments deployed by the trading partner
(as long as imports and exports are correctly measured):
\begin{equation}
Q_{i}\geq1-a_{i}\label{eq:tic_min_Q}
\end{equation}

Assume the agreement specifies that both countries choose the number
of additional certificates such that (\ref{eq:tic_min_Q}) matches
their strategic targets:
\begin{equation}
a_{A}=1-\sigma_{A},\qquad a_{B}=0\label{eq:X_a}
\end{equation}
Then (\ref{eq:tic_min_Q}) and (\ref{eq:X_a}) directly imply a simple
but very robust guarantee:

\begin{proposition}[Robust target protection of TIC]\label{prop:x_guarantee}The
implemented TIC scheme guarantees that each country's strategic target
is satisfied, independent of how the other country might deviate from
the trade agreement, as long as imports and exports are correctly
measured.

\end{proposition}

Robust target protection relies on the fact that the strategic targets
can be implemented in autarky. It does not rule out profitable deviations
by trading partners that increase a country's direct economic costs.

Let $\pi_{i}\geq0$ denote the TIC certificate price in country $i$
formed in a perfectly competitive certificate market, characterized
by the complementary slackness condition
\[
\pi_{i}\left(Q_{i}^{exp}+a_{i}-Q_{i}^{imp}\right)=0.
\]

From the perspective of an importer, the certificate price $\pi_{i}$
works like a positive tariff. For an exporter who can sell one certificate
per exported unit, $\pi_{i}$ works like an effective export subsidy.

\begin{proposition}[Constrained efficiency of TIC agreement]\label{prop:ticpreview}
The resulting allocation under the TIC agreement is constrained efficient
and satisfies $Q_{A}=\sigma_{A}$ and $Q_{B}=2-\sigma_{A}$. The resulting
certificate prices satisfy 
\[
\pi_{A}=\frac{1}{2}\sigma_{A}-m_{0},\qquad\pi_{B}=0,
\]
and correspond to the shadow prices of the strategic constraints (\ref{eq:x_A_constr})
and (\ref{eq:x_B_constr}) in the constrained optimization problem
(\ref{eq:x_ce_problem}). \end{proposition}

\subsection{Implementing the same allocation using tariffs and export subsidies}

A constrained efficient outcome can also be implemented with an agreement
using tariffs and export subsidies only:

\begin{proposition}[Agreement using tariffs and subsidies]\label{prop:no_tic_preview}

Consider the following trade agreement using no TIC. Both countries
set their tariffs and export subsidies equal to the certificate price
that would arise under the TIC agreement:
\[
t_{A}=e_{A}=\pi_{A}=\frac{1}{2}\sigma_{A}-m_{0},\qquad t_{B}=e_{B}=\pi_{B}=0.
\]
This agreement also implements the constrained efficient allocation
and grants both countries the same payoff as the TIC agreement.

\end{proposition}

\subsection{Benefits of the TIC agreement}

While both agreements implement the constrained-efficient allocation,
the TIC agreement has three main benefits over the agreement using
tariffs and subsidies. First, the TIC agreement requires no information
about countries' cost structures while a precise assessment is needed
to design a constrained-efficient agreement using only tariffs and
subsidies. This informational advantage grows in an $n$-country setting
with product-specific strategic importance and different assessments
of trading-partner reliability, as in Section \ref{sec:general_tic}.
Second, automatic certificate-price adjustment allows the TIC agreement
to maintain constrained efficiency as countries' cost structures change,
whereas tariff and subsidy rates may require frequent revision. Third,
the TIC agreement robustly protects each country's strategic target.

We now show that TIC schemes also reduce \textit{incentives }to deviate.
A recurring concern in trade agreements is deviation through hidden
subsidies or non-tariff trade barriers. Assume countries can deviate
from either agreement by increasing their export subsidies $e_{i}$,
by introducing positive production subsidies $s_{i}$ that are paid
to all domestically produced units, or by erecting non-tariff trade
barriers $b_{i}$ that increase the cost of imports into country $i$
in the same way as a tariff of size $b_{i}$ but generate no tariff
income for country $i$.

\begin{proposition}[]\label{prop:deviations_preview} Consider
unilateral deviations in which a country changes only one of the following
instruments: an additional export subsidy, a production subsidy, or
a non-tariff import barrier. No country has a strictly profitable
deviation of this kind from the TIC agreement. By contrast, if $\gamma_{B}>1-\frac{\sigma_{A}}{2}$,
then country $B$ has a strictly profitable deviation from the constrained-efficient
tariff-subsidy agreement with any of these instruments. \end{proposition}

The key intuition is that under the TIC agreement, none of the considered
deviations can increase country $B$'s total production because country
$A$'s TIC system robustly protects its target. Any subsidies or non-tariff
trade barriers by country $B$ will be offset by higher certificate
prices in country $A$. Moreover, the considered deviations are weakly
costly in direct economic terms. For the agreement using tariffs and
subsidies, hidden subsidies or non-tariff trade barriers can indeed
increase country $B$'s production by lowering $A$'s production below
its strategic target. Country $B$ will do so if it has a sufficiently
large preference for higher production.

\subsection{Fiscal preferences for tariffs}

Under both agreements, countries have fiscal incentives to deviate
unilaterally toward higher tariffs. If country $A$ introduces a tariff
in addition to its TIC, it will effectively reduce the certificate
price $\pi_{A}$. The effective export subsidy based on exporters'
income from certificate sales is thus reduced and transformed into
tariff income, which is fiscally more attractive. Yet, such a shift
from certificate prices to tariffs would generate a no-trade zone
and destroy constrained efficiency.

Similarly, countries have a fiscal incentive to modify the TIC scheme
so that the government keeps part or all of the income from certificates
generated by exports. \citet{PapadimitriouEtAl2008} actually propose
that the government should obtain all certificate income in their
numerical analysis of Buffett's original TIC proposal for the U.S.
economy. Yet, even if the government uses the certificate revenue
to reduce general taxes or subsidize domestic production, retaining
the export-generated certificate income generally creates a no-trade
zone, destroying constrained efficiency.

This tension between fiscal incentives and constrained efficiency
also arises in our general framework. On the other hand, enforcing
agreed-upon tariff rates and the agreed distribution of TIC income
may be substantially easier than preventing hidden subsidies or non-tariff
trade barriers \citep{HornMaggiStaiger2010,BeverelliBoffaKeck2019}.

\section{General framework\label{sec:general_tic}}

This section extends the core results of our example to a general
setting. We allow for $N$ countries, a finite set of products, convex
production and transportation costs, intermediate inputs, and elastic
final demand. Each country keeps a single certificate market but can
assign different strategic weights across products and account for
friendshoring relationships with trading partners. Matched TIC mechanisms
continue to robustly protect countries' strategic targets, and TIC-based
trade agreements can implement the constrained-efficient global allocation.
We also show how certificate prices measure the marginal welfare effects
of target parameters and characterize interactions with other trade
instruments such as the EU CBAM.

\subsection{Environment\label{subsec:e_env}}

There is a set $\mathcal{N}=\{1,\ldots,N\}$ of countries indexed
by $i$ and $j$. A customs union, such as the European Union, is
treated as a single country. We consider a finite set of products
$\mathcal{M}=\{1,\ldots,M\}$. Let the flow $q_{ij}(m)\geq0$ denote
the quantity of product $m$ produced in country $i$ and used in
country $j$. We denote by $q$ the complete allocation of production
and trade flows. By country $i$'s \textit{direct} flows we mean its
production flows $q_{ij}(m)$ and import flows $q_{ji}(m)$.

Total production of product $m$ in country $i$ is
\begin{equation}
q_{i}(m)=\sum_{j\in\mathcal{N}}q_{ij}(m),\label{eq:g_total_production}
\end{equation}
while total domestic use of product $m$ is
\begin{equation}
y_{i}(m)=\sum_{j\in\mathcal{N}}q_{ji}(m).\label{eq:g_consumption}
\end{equation}
We distinguish final consumption $y_{i}^{F}(m)\geq0$ from use as
an intermediate input to domestic production $y_{i}^{I}(m)\geq0$.
Let $y^{F}=(y_{1}^{F},\ldots,y_{N}^{F})$ and $y^{I}=(y_{1}^{I},\ldots,y_{N}^{I})$
describe the product use for all countries and denote by
\[
\mathcal{Y}(q)=\left\{ (y^{F},y^{I})\geq0:y_{i}^{F}(m)+y_{i}^{I}(m)=y_{i}(m)\text{ for all }i,m\right\} 
\]
the set of all feasible product uses given allocation $q$.

Each country $i$ has a representative consumer with quasilinear utility
\[
U_{i}(y_{i}^{F},\zeta_{i})=V_{i}(y_{i}^{F})+\zeta_{i},
\]
where $\zeta_{i}$ is a numeraire good with price normalized to one.
We assume that $V_{i}$ is continuously differentiable, nondecreasing,
and concave, with $V_{i}(0)=0$. Profits and policy revenues or expenditures
are rebated or financed through lump-sum transfers.

Worldwide production and transportation costs are represented by a
reduced-form function $C(q,y^{I})$. It describes the primary-resource
costs of producing and delivering $q$ when $y^{I}$ is used as intermediate
inputs, excluding the market value of those inputs. We assume that
$C$ is continuously differentiable, nonnegative, and jointly weakly
convex in $(q,y^{I})$, with $C(0,0)=0$.

We define ordinary world welfare for an allocation $q$ by optimizing
the allocation of domestic use between final consumption and intermediate
use:
\begin{equation}
W(q)=\max_{(y^{F},y^{I})\in\mathcal{Y}(q)}\left\{ \sum_{i\in\mathcal{N}}V_{i}(y_{i}^{F})-C(q,y^{I})\right\} .\label{eq:g_welfare}
\end{equation}
The maximum exists, and the resulting reduced-form welfare function
$W(q)$ is continuous and concave.

Let $\mathcal{Q}=\mathbb{R}_{+}^{N\times N\times M}$ denote the set
of possible flow allocations $q$. To rule out uninteresting cases
of nonexistence caused by unbounded quantities, we assume that the
upper contour set $\{q\in\mathcal{Q}:W(q)\geq0\}$ is compact.\footnote{Since $W(q)$ is continuous in $q$, this condition is satisfied,
for example, if $W(q)\to-\infty$ as $\|q\|\to\infty$.}

\subsection{Strategic targets\label{subsec:g_targets}}

We consider a tractable class of \textit{linear surplus targets} whose
parameters have a clear economic interpretation. These targets are
motivated by our microfoundation in Section \ref{sec:microfoundation},
where they are optimal among a large class of strategic targets.

Country $i$ may regard supply from some trading partners as substantially
more reliable than supply from others. Recent friendshoring policies
explicitly reflect this idea by seeking to preserve the gains from
international specialization while shifting critical supply chains
toward a set of trusted partners \citep{Yellen2022Friendshoring,JavorcikEtAl2024}.
Let 
\begin{equation}
\phi_{ji}(m)\geq0
\end{equation}
 denote the \textit{friendshoring factor} that country $i$ assigns
to one unit of product $m$ imported from country $j\neq i$. A factor
$\phi_{ji}(m)=0$ means that imports of product $m$ from country
$j$ receive no credit toward country $i$'s target, while $\phi_{ji}(m)=1$
means that it is strategically equivalent to securely available domestic
production. Friendshoring factors may also vary across products. A
partner may be considered highly reliable for one supply chain but
less so for another.

A linear surplus target for country $i$ also assigns a strategic
production factor to every product $m$ and destination country $j$:
\begin{equation}
\psi_{ij}(m)\in\mathbb{R}.
\end{equation}

These production factors can account for the other side of friendshoring.
Exports of critical goods to an allied country with a mutual friendshoring
agreement may be less easily redirected to domestic uses in a common
crisis than exports to a geopolitical rival, corresponding to lower
values of $\psi_{ij}(m)$ for the close ally than for the geopolitical
rival. Lower values can also reflect insecure production capability
due to supply chain risk.\footnote{Suppose a product $m$ can be produced domestically using either a
conventional process or a certified, more secure process, for example
with emergency inventories, traceable supply chains, stress-test requirements,
or second-sourcing arrangements. The model can represent such certified
secure output as a new product variant $\tilde{m}$ that is a perfect
substitute for $m$ but gets a higher production factor and has different
production costs.} In this context, a domestic production factor of one is a plausible
benchmark for fully secured domestic production. Section \ref{sec:microfoundation}
illustrates how production factors can further reflect political pressure
from lost exports in a trade crisis and leverage over a geopolitical
rival.

Combining both factors, we define the strategically secured supply
of product $m$ for country $i$ as 
\begin{equation}
\widetilde{q}_{i}(m)=\sum_{j\in\mathcal{N}}\psi_{ij}(m)q_{ij}(m)+\sum_{j\in\mathcal{N}\setminus\{i\}}\phi_{ji}(m)q_{ji}(m).\label{eq:g_tildeq}
\end{equation}
We generally identify a strategic target for country $i$ with a surplus
function $S_{i}(q)$ that must satisfy
\begin{equation}
S_{i}(q)\geq0\label{eq:strat_target_ineq}
\end{equation}
Linear surplus targets are defined by
\begin{equation}
S_{i}(q)=\sum_{m\in\mathcal{M}}w_{i}(m)\left[\widetilde{q}_{i}(m)-\sigma_{i}(m)y_{i}(m)\right],\label{eq:g_strategic_surplus}
\end{equation}
where $\widetilde{q}_{i}(m)-\sigma_{i}(m)y_{i}(m)$ measures the surplus
of secured supply above a specified share $\sigma_{i}(m)\in[0,1]$
of total domestic use $y_{i}(m)$. The \textit{product weight} $w_{i}(m)\geq0$
measures the relative weight country $i$ places on the surplus of
product $m$. 

Our main results are derived generally for targets $S_{i}(q)$ that
are linear in $q$ and depend only on country $i$'s direct flows.
This class includes linear surplus targets. Our characterization does
not extend, however, to affine targets $S_{i}(q)$ that include a
constant term. A constant is required to model an absolute production
target. With elastic demand, constrained-efficient implementation
of absolute production targets using certificates generally requires
the government to buy and surrender a corresponding fixed number of
certificates. We restrict attention to self-financing systems with
certificate obligations for firms only. The exact matching results
below rely on linearity; nonlinear targets would generally require
allocation-dependent certificate coefficients.

\subsection{Matched TIC\label{subsec:g_matched_tic}}

A TIC mechanism for country $i$ is defined by coefficients $a_{jk}^{i}(m)$
that specify the number of certificates received when $a_{jk}^{i}(m)>0$,
or surrendered when $a_{jk}^{i}(m)<0$, per unit of $q_{jk}(m)$.
Only firms responsible for the trade flows receive or must surrender
certificates: producers for $a_{ij}^{i}(m)$ and importers for $a_{ji}^{i}(m)$
with $j\ne i$. If country $i$ implements no TIC, we simply assume
$a^{i}=0$. The certificate surplus of a TIC system is defined by
\begin{equation}
A_{i}(q)=\sum_{m\in\mathcal{M}}\sum_{j\in\mathcal{N}}\sum_{k\in\mathcal{N}}a_{jk}^{i}(m)q_{jk}(m)\label{eq:TIC_A}
\end{equation}
and an allocation $q$ is feasible for a TIC system if and only if
\begin{equation}
A_{i}(q)\geq0.\label{eq:TIC_feasibility}
\end{equation}
We define the \textit{matched TIC} mechanism for a linear target by
\begin{equation}
a_{jk}^{i}(m)=\frac{\partial S_{i}(q)}{\partial q_{jk}(m)}\text{\qquad\text{for all }}j,k\in\mathcal{N},m\in\mathcal{M}\label{eq:matched_tic}
\end{equation}
Thus, for a linear surplus target, the matched TIC assigns to exports
\begin{equation}
a_{ij}^{i}(m)=\frac{\partial S_{i}(q)}{\partial q_{ij}(m)}=\psi_{ij}(m)w_{i}(m)\text{\qquad\text{for all }}j\ne i\label{eq:g_export_credit}
\end{equation}
certificates per unit. For imports, the matched TIC has
\begin{equation}
a_{ji}^{i}(m)=\frac{\partial S_{i}(q)}{\partial q_{ji}(m)}=\bigl[\phi_{ji}(m)-\sigma_{i}(m)\bigr]w_{i}(m)\text{\qquad\text{for all }}j\ne i.\label{eq:g_import_requirement}
\end{equation}
This means importers with a large friendshoring factor satisfying
$\phi_{ji}(m)>\sigma_{i}(m)$ receive certificates, while other importers
must surrender certificates. For domestic producers we find:\footnote{In the special case of perfectly inelastic demand studied in Section
\ref{sec:twocountry}, linear surplus targets can also be implemented
constrained efficiently by TIC schemes in which the state can auction
off certificates instead of granting certificates to producers for
domestic markets.}
\begin{equation}
a_{ii}^{i}(m)=\frac{\partial S_{i}(q)}{\partial q_{ii}(m)}=(\psi_{ii}(m)-\sigma_{i}(m))w_{i}(m).\label{eq:g_domestic_credit}
\end{equation}
Thus secured domestic production with $\psi_{ii}(m)=1$ will receive
certificates whenever $\sigma_{i}(m)<1$. If a high supply-chain risk
lowers the production factor below the target share, $\psi_{ii}(m)<\sigma_{i}(m)$,
the domestic production flow must surrender certificates.

Naturally, for trade flows not to or from country $i$, we have
\[
a_{jk}^{i}(m)=\frac{\partial S_{i}(q)}{\partial q_{jk}(m)}=0\text{\qquad\text{for all }}j,k\ne i
\]

\begin{proposition}\label{prop:g_representation}The matched TIC
mechanism satisfies for a linear target
\[
A_{i}(q)=S_{i}(q)\text{\qquad\text{for all }}q\in\mathcal{Q}.
\]

\end{proposition}

We assume that certificates are traded on perfectly competitive markets
and denote the resulting certificate price by $\pi_{i}$. We formally
define a competitive equilibrium below, assuming that countries use
only TIC as trade instruments. The conditions adapt naturally when
additional trade instruments such as tariffs or subsidies are used.
\begin{defn}
[Competitive equilibrium under TIC] Assume countries implement only
TIC systems as trading instruments. A competitive equilibrium consists
of an allocation $(q,y^{F},y^{I})$, goods prices $p_{i}(m)\geq0$,
and certificate prices $\pi_{i}\geq0$ such that the following conditions
hold. Representative consumers optimize:
\[
y_{i}^{F}\in\arg\max_{\widehat{y}_{i}^{F}\geq0}\left\{ V_{i}(\widehat{y}_{i}^{F})-\sum_{m\in\mathcal{M}}p_{i}(m)\widehat{y}_{i}^{F}(m)\right\} \text{for all }i.
\]
Aggregate production and trading choices solve:
\[
(q,y^{I})\in\arg\max_{\widehat{q},\widehat{y}^{I}\geq0}\left\{ \sum_{i,j\in\mathcal{N}}\sum_{m\in\mathcal{M}}p_{j}(m)\widehat{q}_{ij}(m)-\sum_{i\in\mathcal{N}}\sum_{m\in\mathcal{M}}p_{i}(m)\widehat{y}_{i}^{I}(m)-C(\widehat{q},\widehat{y}^{I})+\sum_{\ell\in\mathcal{N}}\pi_{\ell}A_{\ell}(\widehat{q})\right\} .
\]
Goods markets clear according to
\[
y_{j}^{F}(m)+y_{j}^{I}(m)=\sum_{i\in\mathcal{N}}q_{ij}(m)=y_{j}(m)\qquad\text{for all }j,m,
\]
and each certificate market satisfies $A_{i}(q)\geq0$ and
\begin{equation}
\pi_{i}A_{i}(q)=0.\label{eq:pi_cs}
\end{equation}
\end{defn}

\subsection{Robust target protection\label{subsec:target_protection}}

We specify a broad class of direct policy instruments against which
the matched TIC mechanism provides robust target protection. A direct
policy instrument of country $j$ can affect the costs or revenues
of firms, consumers, or governments associated with any direct flow
of $j$, e.g. through tariffs, subsidies, non-tariff cost wedges,
or certificate schemes. It can also quantitatively restrict direct
flows, for example by banning imports of product $m$ from country
$i$. Direct quantitative restrictions by a country must admit an
allocation in which all of its international trade flows are zero;
in particular, they cannot force another country to trade.

\begin{proposition}[Robust target protection of TIC] \label{prop:rtp}
Suppose country $i$ has a linear target and implements the matched
TIC system, while every other country $j$ can use any direct policy
instrument. Then every resulting competitive equilibrium allocation
$q$ satisfies country $i$'s linear target, i.e. $S_{i}(q)\geq0$.\end{proposition}

The strong robust-target-protection result relies on linear targets
being implementable under autarky in our model. Robust protection
does not rule out inefficient outcomes if trading partners deviate.

\subsection{Constrained efficiency\label{subsec:g_efficiency}}

A constrained-efficient allocation maximizes ordinary world welfare
subject to every country's strategic targets: 
\begin{align}
q^{*}\in\arg\max_{q\in\mathcal{Q}}\quad & W(q)\label{eq:g_planner}\\
\text{s.t.}\quad & S_{i}(q)\geq0, & i\in\mathcal{N} & ,\label{eq:g_constr}
\end{align}

\begin{proposition}[Constrained efficiency of TIC] \label{prop:g_efficiency}Assume
every country has a linear target and implements as its sole trade
instrument the matched TIC system. Then every competitive equilibrium
allocation is constrained efficient, and the certificate price $\pi_{i}$
is a shadow price of country $i$'s strategic constraint (\ref{eq:g_constr}).
A constrained-efficient allocation and a competitive equilibrium exist.
Conversely, every constrained-efficient allocation can be decentralized
as a competitive equilibrium under the matched TIC.

\end{proposition}

A straightforward corollary is that if an efficient free-trade allocation
satisfies all countries' targets, it can be supported as a competitive
equilibrium if all countries use matched TIC systems as their only
trade instruments.

\subsection{Measuring marginal welfare effects of changes in strategic targets\label{sec:cost_targets}}

Certificate prices make the marginal effects of different parameters
of a strategic target on ordinary welfare observable. Assume every
country implements only the matched TIC system, so that a constrained-efficient
allocation $q^{*}$ is implemented. Let $W^{*}$ denote constrained-efficient
ordinary welfare and let $\theta_{i}$ denote any parameter of country
$i$'s strategic target. At parameter values where $W^{*}$ is differentiable,
the envelope theorem gives 
\begin{equation}
\frac{dW^{*}}{d\theta_{i}}=\pi_{i}\frac{\partial S_{i}(q^{*})}{\partial\theta_{i}}.\label{eq:dW_dtheta}
\end{equation}
This expression measures the marginal effect on ordinary world welfare
of a marginal change in that strategic parameter.

For a linear surplus target, the expression specializes as follows.
A change in the target share for product $m$ gives 
\begin{equation}
\frac{dW^{*}}{d\sigma_{i}(m)}=-\pi_{i}w_{i}(m)y_{i}^{*}(m).\label{eq:c_share_shadow}
\end{equation}

Increasing the friendshoring factor assigned to imports of product
$m$ from country $j$ relaxes country $i$'s target in proportion
to the current import flow: 
\begin{equation}
\frac{dW^{*}}{d\phi_{ji}(m)}=\pi_{i}w_{i}(m)q_{ji}^{*}(m).\label{eq:c_phi_shadow}
\end{equation}
Increasing the strategic production factor assigned to exports from
$i$ to $j$ does the same in proportion to the current export flow.
Thus 
\begin{equation}
\frac{dW^{*}}{d\psi_{ij}(m)}=\pi_{i}w_{i}(m)q_{ij}^{*}(m),\qquad j\neq i.\label{eq:c_psi_shadow}
\end{equation}
For the product weight $w_{i}(m)$ we find 
\begin{equation}
\frac{dW^{*}}{dw_{i}(m)}=\pi_{i}\left[\widetilde{q}_{i}^{*}(m)-\sigma_{i}(m)y_{i}^{*}(m)\right].\label{eq:c_weight_shadow}
\end{equation}
The term in brackets is product $m$'s contribution to the strategic
surplus per unit of its weight. Increasing the weight of a product
that is currently abundant in strategic terms relaxes the target;
increasing the weight of a strategically scarce product tightens it.

\subsection{Interaction of TIC with other policy instruments\label{subsec:g_policy_overlap}}

TIC are designed to protect a strategic supply target directly, not
to replace other policy instruments aimed at environmental, innovation,
regional, or other domestic objectives. As an example, consider the
EU Carbon Border Adjustment Mechanism (CBAM). Importers must surrender
CBAM certificates according to the emissions embedded in covered imports,
and certificate prices are tied to EU ETS allowance prices \citep{EUCBAM2023}.

Are such policy instruments more or less problematic in a world where
countries also use matched TIC mechanisms? Assume country $i$ introduces
a CBAM. Country $j\neq i$ may distrust the stated environmental motivation
and suspect that country $i$ uses such a measure to give an unfair
advantage to strategically relevant domestic industries; see \citet{BagwellStaiger2001,Ederington2001,Lee2007}
for the literature on trade agreements and domestic policies. A matched
TIC system in country $j$ should alleviate one part of this concern,
because country $j$'s strategic target is robustly protected as long
as it remains feasible.

But what about the case that country $j$'s strategic constraint is
slack while country $i$'s strategic constraint binds with $\pi_{i}>0$?
Does adding another policy instrument, such as CBAM, in country $i$
lead to excessive strategic protection relative to a pure TIC scheme?
The following result shows that this is not the case.

\begin{proposition}[Strategic non-additivity] \label{prop:g_strategic_nonadditivity}
Suppose country $i$ implements the matched TIC for its linear strategic
target $S_{i}(q)$. Assume the strategic target binds in the initial
competitive equilibrium with a strictly positive certificate price.
Consider the introduction of an additional policy instrument while
country $i$'s matched TIC mechanism remains unchanged. Assume a competitive
equilibrium still exists. Then either country $i$'s strategic surplus
remains zero, or country $i$'s strategic target becomes slack with
a certificate price of zero.

\end{proposition}

The proposition gives a useful sense in which strategic protection
from TIC and other policy instruments does not stack. Yet, even when
a CBAM leaves country $i$'s strategic surplus $S_{i}(q)$ unchanged,
it can still affect the allocation $q$. In line with the environmental
goals of CBAM, imports with higher embedded emissions face relatively
larger costs, preserving incentives to source from cleaner foreign
producers \citep{FischerFox2012,MehlingRitz2023}.

\section{A microfoundation for linear surplus targets\label{sec:microfoundation}}

\subsection{Overview}

Actual security objectives reflect political judgments and geopolitical
contingencies that are difficult to model comprehensively. Section
\ref{sec:general_tic} therefore takes strategic targets as primitive.
This section presents a simple economic model that provides a coherent
microfoundation in which a particular linear surplus target $S^{*}$
can prevent coercion by trading partners in a minimally restrictive
way. Furthermore, $S^{*}$ maximizes welfare among all strategic target
profiles that prevent coercion.

Consider the economic environment from Section \ref{sec:general_tic}.
For every country $i$, there is one geopolitical rival $j\neq i$
that can initiate coercive bargaining by threatening to trigger a
trade crisis. Throughout this section, $i$ denotes the defender and
$j$ the coercer.

The timing is as follows.
\begin{enumerate}
\item \textit{Strategic targets and trade agreement. }Countries state strategic
targets and may sign a trade agreement.
\item \textit{Normal trade policy and long-run adjustment.} Countries implement
their trade policy. They can deviate from the trade agreement using
direct trade instruments like hidden subsidies. Then a market allocation
$q$ arises. We interpret $q$ as a long-run allocation to which productive
capacity and usage patterns have adapted.
\item \textit{Coercive bargaining.} After the long-run allocation and associated
capacities are in place, country $i$'s geopolitical rival $j$ can
initiate coercive bargaining by threatening to trigger a trade crisis.
The bargaining outcome depends on a threat point determined by political
pressures from shortages and lost export markets during the trade
crisis. We assume a political advantage for the defender that depends
on the strength of an informal norm to resist coercion. Coercive bargaining
is initiated only if $j$ predicts that it can extract a positive
concession from $i$.
\end{enumerate}
All crisis parameters below should be interpreted as components of
country $i$'s ex ante security assessment rather than as a complete
model of the true crisis allocation. 

\subsection{The defender's self-assessed crisis exposure}

We first describe the supply available for domestic use that country
$i$ expects to keep in the trade crisis. Let 
\begin{equation}
g_{ki}(m)\in[0,1],\qquad k\neq i,\label{eq:mf_g}
\end{equation}
denote the expected fraction of normal imports $q_{ki}(m)$ from country
$k$ to $i$ that remains available in the crisis. If the coercer
is expected to stop all exports to the defender in the crisis then
$g_{ji}(m)=0$.

Let 
\begin{equation}
d_{ik}(m)\in[0,1],\label{eq:mf_d}
\end{equation}
denote the fraction of the defender's normal production sold in country
$k$ that can be made available for domestic use at short notice.
Expected supply chain interruptions correspond to lower values of
$d_{ik}(m)$. For exports, lower values of $d_{ik}$ also correspond
to committed exports that cannot be easily redirected to domestic
use.

Country $i$'s assessed crisis availability of product $m$ is therefore
\begin{equation}
y_{i}^{c}(m)=\sum_{k\in\mathcal{N}}d_{ik}(m)q_{ik}(m)+\sum_{k\in\mathcal{N}\setminus\{i\}}g_{ki}(m)q_{ki}(m).\label{eq:mf_crisis_consumption_i}
\end{equation}

Let $\omega_{i}^{y}(m)>0$ measure the assessed political salience
in country $i$ of a shortage of product $m$. We define politically
weighted normal use and assessed crisis availability by 
\begin{align}
Y_{i} & =\sum_{m}\omega_{i}^{y}(m)y_{i}(m),\label{eq:mf_Y_i}\\
Y_{i}^{c} & =\sum_{m}\omega_{i}^{y}(m)y_{i}^{c}(m).\label{eq:mf_Yc_i}
\end{align}

The weights $\omega_{i}^{y}(m)$ need not exactly coincide with ordinary
expenditure or marginal-utility weights. A product could account for
modest ordinary expenditure but be highly salient when an abrupt shortage
threatens essential economic functions.

Political pressure may arise not only from domestic shortages but
also from firms and workers when production for export markets drops.
For products where the pressure on country $i$ from lost exports
appears stronger than country $j$'s own pressure from shortages,
country $j$ may therefore deliberately restrict its imports. 

Let $\chi_{ij}(m)\in[0,1]$ denote the fraction of country $i$'s
normal exports of product $m$ to $j$ that is expected not to be
sold during a crisis. Let $\omega_{i}^{x}(m)\geq0$ measure the political
pressure per unit of lost exports. We define the export-market component
of country $i$'s crisis loss as
\begin{equation}
L_{i}^{x}(q)=\sum_{m}\omega_{i}^{x}(m)\chi_{ij}(m)q_{ij}(m).\label{eq:mf_Lix}
\end{equation}

\subsection{The defender's assessment of the coercer's crisis exposure}

Country $i$ generally has less information about country $j$'s true
crisis position. In particular, it may not know which third countries
would continue supplying $j$. We consider a conservative assessment
in which country $i$ assumes that third-country deliveries to $j$
continue normally. Let
\begin{equation}
\eta_{ij}(m)\in[0,1]\label{eq:mf_eta-1}
\end{equation}
denote the defender's assessed fraction of its direct exports of product
$m$ to country $j$ that country $j$ loses and cannot replace during
the trade crisis. Let $\omega_{j}^{y}(m)\geq0$ denote country $i$'s
assessment of the political salience of product $m$ in country $j$.
The resulting shortage component of country $j$'s crisis loss is
\[
L_{j}^{y}(q)=\sum_{m}\omega_{j}^{y}(m)\eta_{ij}(m)q_{ij}(m).
\]

Country $i$ also assesses the coercer's dependence on access to country
$i$'s market. Let $\chi_{ji}(m)\in[0,1]$ denote the assessed fraction
of country $j$'s exports to market $i$ that is lost and cannot be
replaced at short notice, and let $\omega_{j}^{x}(m)\geq0$ denote
their assessed political salience in country $j$. The corresponding
export-market loss is
\begin{equation}
L_{j}^{x}(q)=\sum_{m}\omega_{j}^{x}(m)\chi_{ji}(m)q_{ji}(m).\label{eq:mf_Ljx}
\end{equation}

\subsection{Political pressure, norms against coercion, and the Nash bargaining
solution}

The coercer's domestic political pressure during the trade crisis
is
\begin{equation}
L_{j}(q)=L_{j}^{y}(q)+L_{j}^{x}(q)\geq0.\label{eq:mf_Lj}
\end{equation}

We now make a key assumption that limits the scope for successful
coercion attempts. We assume that there is a social norm against giving
in to coercion. We model this norm by allowing the defender's politically
weighted domestic availability to fall by a share $\alpha_{i}\in[0,1]$
before shortages generate political pressure. We define the defender's
political pressure during the crisis by 
\begin{align}
L_{i}(q) & =\max\{L_{i}^{y}(q)+L_{i}^{x}(q),0\}\qquad\text{with }\label{eq:mf_Pi}\\
L_{i}^{y}(q) & =(1-\alpha_{i})Y_{i}-Y_{i}^{c}
\end{align}
Larger values of $\alpha_{i}$ can be interpreted as stronger public
and institutional support for political leaders to withstand coercive
pressure during a crisis.\footnote{A related interpretation is that a coercive crisis is a strong breach
of the trade agreement. Even if non-coercion is not exogenously enforceable,
the non-coercion outcome can still shape subsequent bargaining positions
and perceived entitlements \citep{EdlinReichelstein1996,GoldlueckeKranz2023,HartMoore2008}.}

Bargaining outcomes are described by the symmetric Nash bargaining
solution. The disagreement payoffs during the crisis are $\left(-L_{i}(q),\,-L_{j}(q)\right)$.\footnote{See e.g. \citet{GrossmanHelpman1995FTA} and \citet{GrossmanHelpman1995TradeWars}
for other trade models with bargaining solutions based on political
targets that are affected by economic fundamentals.} If the countries reach an agreement, the trade crisis is lifted immediately
and the crisis-induced political pressures disappear. Let $x$ denote
the value of the concession made by country $i$ to country $j$,
measured in the same normalized payoff units as political pressure.
We assume perfectly transferable utility, so settlement payoffs are
$(-x,x)$. The Nash bargaining solution with transferable utility
maximizes $(L_{i}-x)(L_{j}+x)$, yielding 
\begin{equation}
x=\tfrac{1}{2}(L_{i}-L_{j}).\label{eq:mf_x}
\end{equation}

\subsection{Strategic targets and main results}

We assume country $j$ initiates coercive bargaining against $i$
only if the bargaining outcome yields a positive concession $x>0$
from country $i$. Whether that is the case depends on the equilibrium
market allocation $q$. Let 
\begin{equation}
\mathcal{P}_{i}=\{q\in\mathcal{Q}:L_{i}(q)\leq L_{j}(q)\}\label{eq:mf_P}
\end{equation}
denote the set of allocations $q$ for which $x\leq0$, so that coercive
bargaining against country $i$ is prevented.

Let $\mathcal{S}_{i}=\{S_{i}\mid S_{i}:\mathcal{Q}\to\mathbb{R}\}$
denote the set of all strategic targets for country $i$. For every
$S_{i}\in\mathcal{S}_{i}$, define
\[
\mathcal{F}_{i}^{*}(S_{i})=\{q\in\mathcal{Q}:S_{i}(q)\geq0\}
\]
as the set of allocations satisfying the target. No functional-form
restriction is imposed on targets in $\mathcal{S}_{i}$; in particular,
they need not be linear or implementable by a matched TIC system.
We say that $S_{i}$ \textit{prevents coercion} if 
\begin{equation}
\mathcal{F}_{i}^{*}(S_{i})\subseteq\mathcal{P}_{i}.\label{eq:mf_prevents_coercion}
\end{equation}

Among strategic targets that prevent coercion, a natural criterion
is to avoid restrictions that provide no additional security. We say
that a strategic target $S_{i}$ prevents coercion in a \textit{minimally
restrictive} way if 
\begin{equation}
\mathcal{F}_{i}^{*}(S_{i})=\mathcal{P}_{i}.\label{eq:mf_min_restrictive}
\end{equation}
\begin{proposition}\label{prop:S_star} The linear surplus target
$S_{i}^{*}$ with parameters 
\begin{align}
w_{i}(m) & =\omega_{i}^{y}(m),\label{eq:mf_map_w}\\
\phi_{ki}(m) & =g_{ki}(m),\qquad k\neq i,j,\label{eq:mf_map_phi_other}\\
\phi_{ji}(m) & =g_{ji}(m)+\frac{\omega_{j}^{x}(m)}{\omega_{i}^{y}(m)}\chi_{ji}(m),\label{eq:mf_map_phi_j}\\
\psi_{ik}(m) & =d_{ik}(m),\qquad k\neq j,\label{eq:mf_map_psi_other}\\
\psi_{ij}(m) & =d_{ij}(m)+\frac{\omega_{j}^{y}(m)}{\omega_{i}^{y}(m)}\eta_{ij}(m)-\frac{\omega_{i}^{x}(m)}{\omega_{i}^{y}(m)}\chi_{ij}(m),\label{eq:mf_map_psi_j}\\
\sigma_{i}(m) & =1-\alpha_{i},\qquad m\in\mathcal{M}.\label{eq:mf_map_s}
\end{align}
prevents coercion in a minimally restrictive way.\end{proposition}

The microfoundation provides a structural interpretation for the parameters
of linear surplus targets. Product weights $w_{i}(m)$ reflect the
political salience of product shortages. For third-country imports,
the friendshoring factor $\phi_{ki}(m)$ is simply the assessed share
that remains available in a crisis. For imports from the coercer,
$\phi_{ji}(m)$ additionally reflects the coercer's dependence on
access to country $i$'s market. Thus even imports from the coercer
can have strategic value when disrupting trade is politically costly
for the supplier itself. More generally, $\phi_{ji}(m)>1$ is possible.

For exports to third countries, $\psi_{ik}(m)$ measures short-run
redirectability toward domestic use. For exports to the coercer, $\psi_{ij}(m)$
also incorporates two opposing bilateral dependencies: leverage from
exports that are hard for $j$ to replace raises the production factor,
while country $i$'s own dependence on the export market lowers it.
Hence $\psi_{ij}(m)$ can exceed one and, in sufficiently asymmetric
cases, can be negative. If $\psi_{ij}(m)$ is negative, the matched
TIC system requires the corresponding exporters to surrender certificates. 

Finally, $\sigma_{i}(m)=1-\alpha_{i}$ links the required secured-supply
share to the norm against yielding to coercion: a stronger norm relaxes
the target. In this sense, stronger informal norms not to yield to
coercion and tighter formal supply security targets are substitutes
for each other.

We close this section with a welfare result that provides further
support for $S^{*}=(S_{1}^{*},\ldots,S_{N}^{*})$ as an attractive
profile of strategic targets.

\begin{proposition}\label{prop:optimal_target} Let 
\begin{equation}
\mathcal{S}^{P}=\left\{ S\in\prod_{i\in\mathcal{N}}\mathcal{S}_{i}:\mathcal{F}_{i}^{*}(S_{i})\subseteq\mathcal{P}_{i}\text{ for every }i\in\mathcal{N},\quad\bigcap_{i\in\mathcal{N}}\mathcal{F}_{i}^{*}(S_{i})\neq\emptyset\right\} \label{eq:mf_SP}
\end{equation}
denote the set of jointly feasible profiles of strategic targets that
prevent coercion in every country. For every $S\in\mathcal{S}^{P}$,
define the highest welfare level compatible with the target profile
by
\begin{equation}
\overline{W}(S)=\sup_{q\in\bigcap_{i\in\mathcal{N}}\mathcal{F}_{i}^{*}(S_{i})}W(q).\label{eq:mf_welfare_sup}
\end{equation}
We have $S^{*}\in\mathcal{S}^{P}$ and $S^{*}$ maximizes welfare
among all target profiles that prevent coercion:
\[
\overline{W}(S^{*})\geq\overline{W}(S)\qquad\text{for every }S\in\mathcal{S}^{P}.
\]
\end{proposition}

Since every $S_{i}^{*}$ is a linear surplus target, welfare-maximizing
allocations under $S^{*}$ can be robustly implemented as competitive
equilibria using matched TIC systems.

\section{Discussion\label{sec:discussion}}

\subsection{WTO rules }

A detailed legal assessment of TIC is beyond the scope of this paper.
At face value, however, the proposed TIC mechanism is likely to be
inconsistent with current WTO rules. Yet, WTO reform discussions now
explicitly cover industrial subsidies, policy space, transparency,
and level-playing-field concerns. The 2026 World Trade Report frames
part of the institutional challenge as preserving ``legitimate policy
space'' while avoiding a shift toward ``unilateral measures or fragmented
approaches'' \citep{WTOReport2026}. The European Union's July 2026
reform submission similarly recognizes that state interventions can
pursue legitimate public-policy objectives while generating negative
spillovers, and calls for work on disciplines, transparency, and remedies
\citep{EUWTOIndustrialPolicy2026}.

Under current WTO rules, GATT Article XI may pose an issue for the
import side of TIC. It generally prohibits quantitative import restrictions
implemented through quotas, licensing requirements, or other measures.
India-Autos provides a particularly close precedent. India required
automobile manufacturers to balance the value of specified imports
with exports, and the panel found that linking permissible imports
to export performance constituted a restriction on importation under
Article XI:1 \citep{WTOIndiaAutos2002}.

One possible modification is a price safety valve. Assume the government
of country $i$ commits to sell unlimited additional certificates
at a fixed price $\bar{\pi}_{i}$. This resembles tariff-rate quotas,
which WTO jurisprudence has treated as outside the prohibition of
Article XI:1 because additional imports remain possible at the higher
tariff \citep{WTOECBananas2008}.\footnote{Article II disciplines on tariffs and other border charges would then
become relevant.} Economically, $\bar{\pi}_{i}$ caps the shadow cost on ordinary welfare
from the strategic target, at the expense of robust target protection
whenever the cap binds.

An exporter's revenue from certificate sales might qualify as a prohibited
export-contingent subsidy under the SCM Agreement \citep{WTO1994SCM}.
A TIC mechanism in which the government, rather than exporters, receives
all certificate income removes this export-side reward. Indeed, countries
have unilateral fiscal incentives to adopt that design. Yet, doing
so destroys constrained efficiency. In our motivating example, such
a TIC mechanism for country $A$ implements the same outcome as an
agreement with $t_{A}=\pi_{A}>e_{A}=0$ and yields an inefficient
no-trade zone.

The possibility that an efficient agreement combines a positive import
tariff with an offsetting export subsidy is not unique to our framework.
In standard perfect-competition terms-of-trade models, cooperative
efficiency can be implemented with equal and offsetting import and
export wedges \citep{Maggi2014,Grossman2016}, while \citet{BagwellStaiger2012}
obtain the same equality in a Cournot delocation model.\footnote{\citet{Collie2000} obtains a rationale for banning strategic export
subsidies under Cournot competition and sufficiently costly public
funds, but the strategic incentives reverse under price competition
\citep{EatonGrossman1986}.}

One reason for strict WTO subsidy rules may be administrability. Subsidies
are unusually opaque and difficult to measure \citep{Gulotty2022}.
The SCM Agreement therefore uses a relatively bright-line prohibition,
including de facto export contingency; the Appellate Body has explained
that this extension was meant to ``prevent circumvention'' \citep{WTOExportContingency2013}.
The legal consequence is important: prohibited subsidies receive accelerated
dispute settlement without the adverse-effects showing required for
actionable subsidies \citep{WTO1994SCM,WTOProhibitedSubsidiesOverview}.
\citet{GreenTrebilcock2007} likewise argue that a simple prohibition
can be preferable to a more flexible standard.

Matched TIC systems address much of this institutional concern. Their
export rewards are not hidden discretionary interventions but follow
mechanically from published target parameters and observable certificate
prices. Robust target protection also reduces the need to police foreign
subsidies merely to prevent them from undermining an agreed strategic
target. These features make matched TIC a natural candidate for consideration
in WTO reform discussions.

More generally, WTO scrutiny could move upstream from the implementing
instrument to the strategic target. Target shares, product weights,
and bilateral friendshoring factors should be disciplined where they
disguise protectionism or unjustifiably discriminate among trading
partners. Once a target is accepted as legitimate, however, international
rules should not force countries away from its constrained-efficient
implementation absent an additional cross-border harm. This target-centered
approach may offer a useful principle for current WTO reform debates
over economic security and state intervention.

\subsection{Related certificate trading systems}

The most economically important applications of tradeable permits
today are emissions trading systems such as the EU ETS. Our constrained-efficiency
result has a direct counterpart in this literature. \citet{Dales1968}
and \citet{Montgomery1972} show that, conditional on an aggregate
emissions constraint, competitive permit trading implements the least-cost
allocation, while the permit price reflects the shadow value of relaxing
the constraint.

The non-additivity result from Subsection \ref{subsec:g_policy_overlap}
has an opposite appeal in emissions and strategic trade policy. \citet{HerwegSchmidt2022}
illustrate the ``waterbed effect'' of cap-and-trade: under a binding
cap, additional voluntary abatement does not reduce aggregate emissions.
For strategic supply, the analogous non-additivity seems desirable.
Trading partners may value the assurance that additional domestic
policies cannot raise strategic surplus beyond the agreed target while
TIC remain binding.

More closely related to TIC with endogenously generated certificates
are tradeable green certificates and tradeable performance standards.
Eligible renewable generation can create certificates that obligated
suppliers must acquire, while performance standards assign credits
or obligations relative to a benchmark \citep{AmundsenMortensen2001}.
Like TIC, such systems use an endogenous certificate price to generate
differentiated rewards and charges. TIC apply this logic to international
strategic-supply targets, with certificate coefficients that can vary
by product, origin, and destination.

\subsection{Focal points in international agreements: prices versus quantities}

In climate policy, \citet{Weitzman2014} and \citet{SchmidtOckenfels2021}
argue that a common carbon price can provide a simpler focal point
than a vector of national quantity commitments. Trade negotiations
may instead favor quantity targets as focal points. Strategic targets
can be transparent and symmetrically structured, while the vector
of tariffs and subsidies that implements them may be highly asymmetric
and depend on technologies and trade responses that are difficult
to observe or forecast. TIC allow countries to negotiate the former
and let certificate markets determine the latter. This advantage is
strongest for relatively coarse rules; extensive differentiation of
product weights and friendshoring factors can itself make the target
system difficult to negotiate.

\subsection{Market power in certificate markets and the case for broad targets\label{subsec:market_power}}

New policy instruments are often introduced with a narrow scope to
build experience. Unfortunately, it is a bad idea to introduce TIC
initially only for a small sector such as rare-earth minerals because
this can lead to domestic producers with market power in the certificate
market. Such a producer could limit domestic production and exports
to raise the certificate price. A higher certificate price increases
revenue per certificate sold and provides stronger protection of the
domestic market by increasing the effective tariff.\footnote{A formalization of this insight for a simple oligopoly model in the
certificate market can be found in the predecessor of this paper \citet{Kranz2025TIC}.}

Our strong recommendation is to introduce TIC only with sufficiently
broad targets that cover enough markets so that no single producer
can meaningfully affect certificate prices by restricting supply.
Higher friendshoring factors for countries with producers competing
with domestic producers can also help limit market power in the certificate
market.

\subsection{Transformative AI}

The possibility of transformative AI has been a very important motivation
for this paper. Historically, institutions that foster economic efficiency
and international competition have been crucial for productivity growth
and rising living standards. However, very large AI-driven productivity
gains may shift political priorities away from maximizing ordinary
economic surplus. Voters and governments may instead place greater
weight on the distribution of those gains, fiscal capacity, human
control over important decisions, and the preservation of meaningful
roles for humans \citep{KorinekLockwood2026}.

Unconstrained international competition could make it harder to achieve
such objectives. While by no means certain, AI and advanced robotics
may eventually be able to automate a large share of work in strategically
important sectors \citep{TrammellKorinek2026}. If production becomes
much less dependent on human labor, it could become much more mobile
across countries. Taxing or regulating AI in one country may induce
relocation toward jurisdictions with less stringent policies. This
concern is familiar from international tax competition \citep{KatoLoebbing2023}.
To preserve domestic policy space, broader strategic supply targets
could therefore become a natural response by governments.

Moreover, rapid technological change could alter comparative advantage,
production costs, and the location of production much faster than
fixed tariff or subsidy schedules can be renegotiated. TIC, with certificate
prices that automatically adjust to changing economic conditions,
may become particularly attractive.

It is far from certain that the world will soon resemble these transformative
AI scenarios. But we should plan for them.

\bibliographystyle{apalike}
\bibliography{gtic}

\appendix

\section{Proofs}

Subsections A.1, A.2, and A.3 contain the proofs for all results in
Sections \ref{sec:twocountry}, \ref{sec:general_tic}, and \ref{sec:microfoundation},
respectively.

\subsection*{A.1 Proofs for Section 2}

Throughout, product markets are competitive with constant marginal
costs, so each market is served by the supplier with the lower delivered
cost. Let $d\in[0,1]$ denote the measure of products that country
$A$ supplies in its own market and $z\in[0,1]$ the measure it exports
to $B$, so that $Q_{A}=d+z$ and $Q_{B}=2-d-z$. Direct cost $D_{i}$
is consumer expenditure minus net government revenue from all trade
instruments (tariffs, subsidies, certificate auctions). Let $F(x)=x^{2}/2-m_{0}x$
and $[x]_{0}^{1}=\max\{0,\min\{1,x\}\}$.

\begin{lemma}[Cutoffs and direct costs]\label{lem:margins} Suppose
countries use only tariffs $t_{i}\ge0$ and export subsidies $e_{i}\ge0$.
Then 
\[
d=[m_{0}+t_{A}-e_{B}]_{0}^{1},\qquad z=[m_{0}-t_{B}+e_{A}]_{0}^{1},
\]
and, up to constants that do not depend on any policy, 
\[
D_{A}=F(d)-e_{B}(1-d)+e_{A}z,\qquad D_{B}=F(z)-e_{A}z+e_{B}(1-d).
\]
\end{lemma}
\begin{proof}
$A$ serves its home market for product $m$ iff $c_{A}(m)\le c_{B}(m)+t_{A}-e_{B}$,
i.e.\ $m\le m_{0}+t_{A}-e_{B}$, and $B$'s market iff $c_{A}(m)+t_{B}-e_{A}\le c_{B}(m)$.
Consumer expenditure in $A$ is $\int_{0}^{d}c_{A}+\int_{d}^{1}(c_{B}+t_{A}-e_{B})$,
tariff revenue is $t_{A}(1-d)$ and subsidy outlays are $e_{A}z$;
substituting $c_{A}-c_{B}=m-m_{0}$ yields $D_{A}$ up to the constant
$\int_{0}^{1}c_{B}$. The expression for $D_{B}$ is symmetric.
\end{proof}
\begin{proof}[Proof of Proposition \ref{prop:nashpreview}] Let
$d\in[0,1]$ denote the measure of products that country $A$ supplies
in its own market and $z\in[0,1]$ the measure it exports to $B$.
By Lemma \ref{lem:margins}, 
\[
d=[m_{0}+t_{A}-e_{B}]_{0}^{1},\qquad z=[m_{0}-t_{B}+e_{A}]_{0}^{1}.
\]
Given the other country's policies, both countries' optimization problems
in $(d,z)$ are concave. Country $B$'s optimality conditions imply
\[
t_{B}=\gamma_{B}\quad\text{if }z>0,\qquad z=0\quad\text{if and only if }\gamma_{B}\geq m_{0}+e_{A},
\]
and, whenever $d<1$, 
\[
e_{B}=\max\{0,\gamma_{B}+d-1\}.
\]

We first show that $d+z=\sigma_{A}$ in every equilibrium. If $d+z>\sigma_{A}$,
at least one of $t_{A}$ and $e_{A}$ is positive, since otherwise
$d+z\leq2m_{0}<\sigma_{A}$. Reducing a positive instrument then lowers
direct cost without reducing strategic utility. If $d+z<\sigma_{A}$
and $d>0$, increasing $d$ has marginal direct cost 
\[
t_{A}=d-m_{0}+e_{B}<\sigma_{A}+\gamma_{B}<\gamma_{A},
\]
and is therefore profitable. If $d=0$, country $B$'s corner condition
implies $\gamma_{B}\geq1+e_{B}$, so country $A$ can open the domestic
margin at marginal cost at most $e_{B}-m_{0}\leq\gamma_{B}-1-m_{0}<\gamma_{A}$.
Hence 
\[
d+z=\sigma_{A}.
\]

Moreover, $d>0$ and $t_{A}>0$. If $d=0$, then $z=\sigma_{A}>0$,
so $t_{B}=\gamma_{B}$ and $e_{A}=\sigma_{A}-m_{0}+\gamma_{B}$. Shifting
a marginal unit of production from exports to the home market then
saves $z+e_{A}=2\sigma_{A}-m_{0}+\gamma_{B}$ while the cost of opening
the domestic margin is at most $e_{B}-m_{0}\leq\gamma_{B}-1-m_{0}$,
a contradiction. If $t_{A}=0$, then $d\leq m_{0}$ and hence $z=\sigma_{A}-d>m_{0}$;
shifting production from exports to the domestic market is again strictly
profitable.

Suppose first that $z>0$. Then $t_{B}=\gamma_{B}$. We claim that
$e_{A}>0$. Otherwise $z=m_{0}-\gamma_{B}$, and optimality of country
$A$ at the corner $e_{A}=0$ requires $t_{A}\leq z$. Using $d+z=\sigma_{A}$
and $t_{A}=d-m_{0}+e_{B}$ gives 
\[
\gamma_{B}\leq\frac{3m_{0}-\sigma_{A}}{2}.
\]
But $m_{0}<1/2$ implies 
\[
\frac{3m_{0}-\sigma_{A}}{2}<\frac{3}{4}-\frac{\sigma_{A}}{2}<\gamma_{B},
\]
a contradiction. Thus $e_{A}>0$, and country $A$'s optimality condition
along $d+z=\sigma_{A}$ is 
\[
t_{A}=z+e_{A}.
\]

We next claim that $e_{B}>0$. If $e_{B}=0$, the cutoff equations
and the preceding optimality condition imply 
\[
d=2z+\gamma_{B}.
\]
Together with $d+z=\sigma_{A}$, this gives $d=(2\sigma_{A}+\gamma_{B})/3$.
But $e_{B}=0$ requires $\gamma_{B}+d-1\leq0$, or 
\[
\gamma_{B}\leq\frac{3}{4}-\frac{\sigma_{A}}{2},
\]
contrary to the assumption. Hence 
\[
e_{B}=\gamma_{B}+d-1>0.
\]
Substituting the cutoff equations into $t_{A}=z+e_{A}$ now yields
\[
d-z=\frac{1}{2}.
\]
Thus exactly one half of all products are internationally traded.
Furthermore, 
\[
t_{A}=z+e_{A}>e_{A}>0,\qquad t_{B}-e_{B}=1-d>0.
\]

It remains to consider $z=0$. Then $d=\sigma_{A}>1/2$, so only the
products $m>d$ are internationally traded and their measure is $1-\sigma_{A}<1/2$.
Moreover, 
\[
e_{B}=\gamma_{B}+\sigma_{A}-1>0,
\]
where positivity follows from $\gamma_{B}>3/4-\sigma_{A}/2$ and $\sigma_{A}>1/2$.
Country $B$'s corner condition gives $e_{A}\leq\gamma_{B}-m_{0}$,
while 
\[
t_{A}=\sigma_{A}-m_{0}+e_{B},
\]
so 
\[
t_{A}-e_{A}\geq2\sigma_{A}-1>0.
\]
Optimality of $z=0$ for country $A$ also requires $t_{B}-m_{0}\geq t_{A}$,
and hence 
\[
t_{B}\geq\sigma_{A}+e_{B}>e_{B}.
\]

Finally, an equilibrium with $z>0$ exists. Set 
\[
d=\frac{\sigma_{A}}{2}+\frac{1}{4},\qquad z=\frac{\sigma_{A}}{2}-\frac{1}{4},
\]
and choose 
\[
t_{B}=\gamma_{B},\qquad e_{B}=\gamma_{B}+d-1,\qquad e_{A}=z-m_{0}+\gamma_{B},\qquad t_{A}=z+e_{A}.
\]
The assumptions imply $z>0$, $e_{A}>0$, and $e_{B}>0$. The displayed
conditions are exactly the first-order conditions of the two concave
optimization problems, and $t_{A}<\sigma_{A}+\gamma_{B}<\gamma_{A}$
ensures that country $A$ does not prefer to fall short of its target.
Hence these policies form a Nash equilibrium.

The constrained-efficient allocation has $d=z=\sigma_{A}/2$. In every
Nash equilibrium derived above either $d-z=1/2$ or $z=0$ and $d=\sigma_{A}$,
so no Nash equilibrium is constrained efficient. \end{proof}

\begin{proof}[Proof of Proposition \ref{prop:x_guarantee}]
With inelastic unit demand, $Q_{i}=1+Q_{i}^{exp}-Q_{i}^{imp}$, so
the certificate constraint $Q_{i}^{imp}\le Q_{i}^{exp}+a_{i}$ is
equivalent to $Q_{i}\ge1-a_{i}$, whatever prices, subsidies or barriers
prevail. With $a_{A}=1-\sigma_{A}$ and $a_{B}=0$ this is $Q_{A}\ge\sigma_{A}$
and $Q_{B}\ge1$.
\end{proof}
\begin{proof}[Proof of Proposition \ref{prop:ticpreview}]
Let $S_{A}(q)=Q_{A}-\sigma_{A}$ and $S_{B}(q)=Q_{B}-1$. By the
identity in the proof of Proposition \ref{prop:x_guarantee}, $S_{i}(q)$
is the net creation of country-$i$ certificates, so certificate-market
equilibrium means $S_{i}(q)\ge0$, $\pi_{i}\ge0$, $\pi_{i}S_{i}(q)=0$.
For product $m$ and market $A$, the delivered costs inclusive of
certificate payments are $c_{A}(m)$ for $A$ and $c_{B}(m)+\pi_{A}-\pi_{B}$
for $B$; in market $B$ they are $c_{A}(m)-\pi_{A}+\pi_{B}$ and
$c_{B}(m)$. Hence the competitive allocation assigns every unit to
the producer that maximizes $-c_{i}(m)+\pi_{i}$, i.e.\ it maximizes
$L_{\pi}(q)=W(q)+\pi_{A}S_{A}(q)+\pi_{B}S_{B}(q)$ over $\mathcal{Q}$.
For any $\hat{q}$ satisfying both targets, 
\[
W(q)=L_{\pi}(q)\ge L_{\pi}(\hat{q})\ge W(\hat{q}),
\]
so $q$ is constrained efficient. Since $c_{A}-c_{B}$ is increasing
in $m$, the constrained optimum lets $A$ serve both markets for
$m<\sigma_{A}/2$, giving $Q_{A}=\sigma_{A}$ and $Q_{B}=2-\sigma_{A}>1$;
thus $\pi_{B}=0$, and indifference at $m=\sigma_{A}/2$ gives $\pi_{A}=\sigma_{A}/2-m_{0}$.
The pair $(\pi_{A},\pi_{B})$ satisfies the Kuhn--Tucker conditions
of the constrained problem and is therefore its vector of shadow prices.
\end{proof}
\begin{proof}[Proof of Proposition \ref{prop:no_tic_preview}]
Let $p^{*}=\sigma_{A}/2-m_{0}$. Under $t_{A}=e_{A}=p^{*}$, $t_{B}=e_{B}=0$,
Lemma \ref{lem:margins} gives $d=z=\sigma_{A}/2$, the allocation
of Proposition \ref{prop:ticpreview}, and consumer prices coincide
with those under TIC (in $A$: $\min\{c_{A},c_{B}+p^{*}\}$; in $B$:
$\min\{c_{A}-p^{*},c_{B}\}$). Net government revenue of $A$ is $p^{*}(1-\sigma_{A}/2)-p^{*}\sigma_{A}/2=p^{*}(1-\sigma_{A})$
under the tariff-subsidy agreement and $\pi_{A}a_{A}=p^{*}(1-\sigma_{A})$
from the certificate auction under TIC; $B$ has no revenue under
either. Hence $D_{A}$, $D_{B}$, $Q_{A}$ and $Q_{B}$, and therefore
both payoffs, coincide.
\end{proof}
\begin{proof}[Proof of Proposition \ref{prop:deviations_preview}]
\emph{TIC agreement.} Let $\pi_{A}\ge0$ be $A$'s certificate price.
In the ranges considered first, $B$'s certificate constraint is slack;
the cases in which it binds are treated separately below. With additional
export subsidies $e_{i}$, production subsidies $s_{i}$ and non-tariff
barriers $b_{i}$, the delivered costs are $c_{A}-s_{A}$ and $c_{B}-s_{B}-e_{B}+\pi_{A}+b_{A}$
in market $A$, and $c_{A}-s_{A}-e_{A}-\pi_{A}+b_{B}$ and $c_{B}-s_{B}$
in market $B$, so 
\[
d=m_{0}+\pi_{A}+b_{A}+s_{A}-s_{B}-e_{B},\qquad z=m_{0}+\pi_{A}+s_{A}+e_{A}-s_{B}-b_{B}.
\]
Since $A$'s TIC remains in force, $Q_{A}=d+z\ge\sigma_{A}$ and thus
$Q_{B}\le2-\sigma_{A}$: no deviation improves either country's strategic
term. While $\pi_{A}>0$ we have $d+z=\sigma_{A}$, and collecting
consumer expenditure, subsidy outlays and auction revenue $\pi_{A}(1-\sigma_{A})$
as in Lemma \ref{lem:margins} gives, up to constants, 
\begin{align*}
D_{A} & =F(d)+(\pi_{A}+s_{A}+e_{A})z+(b_{A}-s_{B}-e_{B})(1-d),\\
D_{B} & =F(z)+(b_{B}-s_{A}-e_{A}-\pi_{A})z+(e_{B}+s_{B})(1-d).
\end{align*}
Let $p^{*}=\sigma_{A}/2-m_{0}$. Consider one instrument at a time,
the others being zero; the binding certificate constraint pins $\pi_{A}$.
\begin{itemize}
\item $e_{A}$: $\pi_{A}=p^{*}-e_{A}/2$, $d=(\sigma_{A}-e_{A})/2$, $z=(\sigma_{A}+e_{A})/2$,
so $\Delta D_{A}=\frac{e_{A}}{8}(2\sigma_{A}+3e_{A})>0$.
\item $b_{A}$: $\pi_{A}=p^{*}-b_{A}/2$, $d=(\sigma_{A}+b_{A})/2$, $z=(\sigma_{A}-b_{A})/2$,
so $\Delta D_{A}=\frac{b_{A}}{8}(8-6\sigma_{A}-b_{A})>0$.
\item $s_{A}$: $\pi_{A}=p^{*}-s_{A}$, $d=z=\sigma_{A}/2$: the allocation
and $\pi_{A}+s_{A}$ are unchanged, so $\Delta D_{A}=0$.
\item $e_{B}$: $\pi_{A}=p^{*}+e_{B}/2$, $d=(\sigma_{A}-e_{B})/2$, $z=(\sigma_{A}+e_{B})/2$,
so $\Delta D_{B}=\frac{e_{B}}{8}(8-6\sigma_{A}+3e_{B})>0$.
\item $b_{B}$: $\pi_{A}=p^{*}+b_{B}/2$, $d=(\sigma_{A}+b_{B})/2$, $z=(\sigma_{A}-b_{B})/2$,
so $\Delta D_{B}=\frac{b_{B}}{8}(2\sigma_{A}-b_{B})>0$.
\item $s_{B}$: $\pi_{A}=p^{*}+s_{B}$, $d=z=\sigma_{A}/2$, so $\Delta D_{B}=(1-\sigma_{A})s_{B}>0$.
\end{itemize}
These expressions cover the ranges in which $A$'s TIC binds and both
trade margins are interior: $e_{A},b_{A}\le2p^{*}$, $s_{A}\le p^{*}$,
and $e_{B},b_{B}\le\sigma_{A}$. It remains to check larger deviations.
For $A$, once $\pi_{A}=0$, an export subsidy satisfies $d=m_{0},\ z=m_{0}+e_{A}$
until $B$'s TIC binds at $e_{A}=1-2m_{0}$; over this interval $dD_{A}/de_{A}=m_{0}+2e_{A}>0$.
While $B$'s TIC binds and $e_{A}\le1$, $d=(1-e_{A})/2$ and $dD_{A}/de_{A}=(3e_{A}+1)/4>0$;
for $e_{A}\ge1$ the allocation is $(d,z)=(0,1)$. The corner conditions
imply $m_{0}\le\pi_{B}\le m_{0}+e_{A}-1$, so $D_{A}=e_{A}-\pi_{B}\ge1-m_{0}$
up to the same constant, which is its value at $e_{A}=1$. For $b_{A}$,
after $\pi_{A}$ reaches zero, $dD_{A}/db_{A}=1-m_{0}-b_{A}>0$ until
$B$'s TIC binds, and $dD_{A}/db_{A}=(1-b_{A})/4\ge0$ while it binds;
after $b_{A}=1$ imports are zero and $D_{A}$ is constant. For $s_{A}>p^{*}$,
$d=z=m_{0}+s_{A}$ and $dD_{A}/ds_{A}=m_{0}+3s_{A}>0$ until $Q_{A}=1$;
thereafter $d=z=1/2$ and the higher subsidy is offset by $B$'s certificate
price, so $D_{A}$ is constant. For $B$, $A$'s TIC remains binding
for all larger deviations. If $e_{B}\ge\sigma_{A}$, then $(d,z)=(0,\sigma_{A})$
and $D_{B}$ rises one-for-one with $e_{B}$; if $b_{B}\ge\sigma_{A}$,
then $(d,z)=(\sigma_{A},0)$ and $D_{B}$ is constant at a level strictly
above baseline. Finally, for every $s_{B}\ge0$, $\Delta D_{B}=(1-\sigma_{A})s_{B}\ge0$.
Hence no country has a strictly profitable deviation, and $s_{A}\le p^{*}$
is exactly neutral.

\emph{Tariff-subsidy agreement.} With $t_{A}=e_{A}=p^{*}$, $t_{B}=e_{B}=0$
and no TIC, a deviation by $B$ gives $d=\sigma_{A}/2-e_{B}-s_{B}$,
$z=\sigma_{A}/2-s_{B}-b_{B}$, $Q_{B}=2-\sigma_{A}+e_{B}+2s_{B}+b_{B}$
and $D_{B}=\text{const}+F(z)+(b_{B}-p^{*})z+(e_{B}+s_{B})(1-d)$.
Differentiating at zero gives 
\[
\frac{\partial u_{B}}{\partial e_{B}}=\gamma_{B}-\Bigl(1-\frac{\sigma_{A}}{2}\Bigr),\qquad\frac{\partial u_{B}}{\partial s_{B}}=2\gamma_{B}-\Bigl(1-\frac{\sigma_{A}}{2}\Bigr),\qquad\frac{\partial u_{B}}{\partial b_{B}}=\gamma_{B}-\frac{\sigma_{A}}{2},
\]
which are all strictly positive when $\gamma_{B}>1-\sigma_{A}/2$
(using $\sigma_{A}<1$). Each deviation lowers $Q_{A}=d+z$ below
$\sigma_{A}$.
\end{proof}

\subsection*{A.2 Proofs for Section 3}

\begin{proof}[Proof of Proposition \ref{prop:g_representation}]
Since $S_{i}$ is linear and the matched coefficients are given by
(\ref{eq:matched_tic}), 
\[
A_{i}(q)=\sum_{m\in\mathcal{M}}\sum_{j,k\in\mathcal{N}}\frac{\partial S_{i}(q)}{\partial q_{jk}(m)}q_{jk}(m)=S_{i}(q).
\]
\end{proof}

\begin{proof}[Proof of Proposition \ref{prop:rtp}] Every resulting
competitive equilibrium must satisfy country $i$'s certificate-balance
condition $A_{i}(q)\geq0$. By Proposition \ref{prop:g_representation},
$A_{i}(q)=S_{i}(q)$, independently of the additional direct policy
instruments. Hence $S_{i}(q)\geq0$. \end{proof}

\begin{proof}[Proof of Proposition \ref{prop:g_efficiency}] Let
$(q,y^{F},y^{I},p,\pi)$ be a competitive equilibrium. Adding consumer
optimality and production-trading optimality, and using goods-market
clearing, shows that for every $(\widehat{q},\widehat{y}^{F},\widehat{y}^{I})$
satisfying $\widehat{y}_{j}^{F}(m)+\widehat{y}_{j}^{I}(m)=\sum_{i}\widehat{q}_{ij}(m)$,
\[
\sum_{i}V_{i}(y_{i}^{F})-C(q,y^{I})+\sum_{i}\pi_{i}A_{i}(q)\geq\sum_{i}V_{i}(\widehat{y}_{i}^{F})-C(\widehat{q},\widehat{y}^{I})+\sum_{i}\pi_{i}A_{i}(\widehat{q}).
\]
Taking $\widehat{q}=q$ first shows that the equilibrium split between
final and intermediate use attains $W(q)$. Maximizing the right-hand
side over use splits compatible with $\widehat{q}$ therefore gives
\[
W(q)+\sum_{i}\pi_{i}A_{i}(q)\geq W(\widehat{q})+\sum_{i}\pi_{i}A_{i}(\widehat{q})\qquad\text{for every }\widehat{q}\in\mathcal{Q}.
\]
By Proposition \ref{prop:g_representation}, $A_{i}(q)=S_{i}(q)$.
Certificate-market clearing gives $S_{i}(q)\geq0$ and $\pi_{i}S_{i}(q)=0$.
Hence, for every $\widehat{q}$ satisfying all strategic targets,
\[
W(q)=W(q)+\sum_{i}\pi_{i}S_{i}(q)\geq W(\widehat{q})+\sum_{i}\pi_{i}S_{i}(\widehat{q})\geq W(\widehat{q}).
\]
Thus $q$ is constrained efficient, and $\pi_{i}$ is a Lagrange multiplier,
hence a shadow price, of constraint (\ref{eq:g_constr}).

A constrained-efficient allocation exists because $q=0$ is feasible
with $W(0)=0$, while the feasible set intersected with the compact
upper contour set $\{q\in\mathcal{Q}:W(q)\geq0\}$ is closed and compact.

Conversely, let $q^{*}$ be any constrained-efficient allocation and
choose $(y^{F*},y^{I*})$ attaining $W(q^{*})$ in (\ref{eq:g_welfare}).
Then $(q^{*},y^{F*},y^{I*})$ solves the equivalent extended planner
problem that maximizes $\sum_{i}V_{i}(y_{i}^{F})-C(q,y^{I})$ subject
to goods-market balance, nonnegativity, and all strategic constraints.
Its feasible set is polyhedral because all these constraints are linear.
Therefore the first-order optimality conditions admit multipliers
$p_{i}(m)$ for goods-market balance and $\pi_{i}\geq0$ for the strategic
constraints, without an additional constraint qualification. The corresponding
Lagrangian is
\[
\sum_{i}V_{i}(y_{i}^{F})-C(q,y^{I})+\sum_{i,m}p_{i}(m)\left[\sum_{j}q_{ji}(m)-y_{i}^{F}(m)-y_{i}^{I}(m)\right]+\sum_{i}\pi_{i}S_{i}(q).
\]
Its maximization over $y^{F}$ gives the consumer problems, while
maximization over $(q,y^{I})$ gives the production and trading problem
in the equilibrium definition. Since $V_{i}$ is nondecreasing, the
supporting goods prices satisfy $p_{i}(m)\geq0$. Proposition \ref{prop:g_representation}
and complementary slackness give certificate-market feasibility and
clearing. Thus every constrained-efficient allocation can be decentralized
as a competitive equilibrium under matched TIC. Since a constrained-efficient
allocation exists, a competitive equilibrium exists as well. \end{proof}

\begin{proof}[Proof of Proposition \ref{prop:g_strategic_nonadditivity}]
Since $\pi_{i}^{0}>0$, certificate-market complementary slackness
and Proposition \ref{prop:g_representation} imply $S_{i}(q^{0})=0$.
After the additional policies are introduced, the unchanged matched
TIC still requires 
\[
S_{i}(q^{1})\geq0,\qquad\pi_{i}^{1}S_{i}(q^{1})=0.
\]
Hence either $S_{i}(q^{1})=S_{i}(q^{0})=0$, or $S_{i}(q^{1})>0$
and $\pi_{i}^{1}=0$. In the latter case all terms generated by country
$i$'s TIC vanish from the production and trading problem, so removing
that TIC leaves $q^{1}$ as an equilibrium allocation under the remaining
policies. \end{proof}

\subsection*{A.3 Proofs for Section 4}

\begin{proof}[Proof of Proposition \ref{prop:S_star}] Coercive
bargaining is prevented if and only if 
\[
L_{i}^{y}(q)+L_{i}^{x}(q)\leq L_{j}^{y}(q)+L_{j}^{x}(q),
\]
because $L_{j}(q)\geq0$ and $L_{i}(q)=\max\{L_{i}^{y}(q)+L_{i}^{x}(q),0\}$.
Substituting the definitions of crisis availability and the four loss
components, and collecting terms gives 
\begin{align*}
 & \sum_{m\in\mathcal{M}}\omega_{i}^{y}(m)\left[\sum_{k\in\mathcal{N}}d_{ik}(m)q_{ik}(m)+\sum_{k\in\mathcal{N}\setminus\{i\}}g_{ki}(m)q_{ki}(m)-(1-\alpha_{i})y_{i}(m)\right]\\
 & \quad+\sum_{m\in\mathcal{M}}\left[\omega_{j}^{y}(m)\eta_{ij}(m)-\omega_{i}^{x}(m)\chi_{ij}(m)\right]q_{ij}(m)\\
 & \quad+\sum_{m\in\mathcal{M}}\omega_{j}^{x}(m)\chi_{ji}(m)q_{ji}(m)\geq0.
\end{align*}
Using (\ref{eq:mf_map_w})--(\ref{eq:mf_map_s}), this inequality
is exactly the linear surplus target $S_{i}^{*}$. Hence 
\[
\mathcal{F}_{i}^{*}(S_{i}^{*})=\mathcal{P}_{i}.
\]

Now consider any strategic target $S_{i}\in\mathcal{S}_{i}$ that
prevents coercion. By definition, 
\[
\mathcal{F}_{i}^{*}(S_{i})\subseteq\mathcal{P}_{i}.
\]
Since $\mathcal{P}_{i}=\mathcal{F}_{i}^{*}(S_{i}^{*})$, it follows
immediately that 
\[
\mathcal{F}_{i}^{*}(S_{i})\subseteq\mathcal{F}_{i}^{*}(S_{i}^{*}).
\]
\end{proof}

\begin{proof}[Proof of Proposition \ref{prop:optimal_target}]
By Proposition \ref{prop:S_star}, $\mathcal{F}_{i}^{*}(S_{i}^{*})=\mathcal{P}_{i}$
for every country $i$. Moreover, $q=0$ belongs to $\mathcal{P}_{i}$
for every $i$ because all crisis-loss terms are zero at $q=0$. Hence
\[
\bigcap_{i\in\mathcal{N}}\mathcal{F}_{i}^{*}(S_{i}^{*})\neq\emptyset,
\]
so $S^{*}\in\mathcal{S}^{P}$.

Now let $S\in\mathcal{S}^{P}$. Since every $S_{i}$ prevents coercion,
\[
\mathcal{F}_{i}^{*}(S_{i})\subseteq\mathcal{P}_{i}=\mathcal{F}_{i}^{*}(S_{i}^{*})
\]
for every country $i$. Therefore
\begin{equation}
\bigcap_{i\in\mathcal{N}}\mathcal{F}_{i}^{*}(S_{i})\subseteq\bigcap_{i\in\mathcal{N}}\mathcal{F}_{i}^{*}(S_{i}^{*}).\label{eq:mf_joint_inclusion}
\end{equation}
Taking the supremum of $W$ over the two sets gives $\overline{W}(S)\leq\overline{W}(S^{*})$.

It remains to show that the supremum under $S^{*}$ is attained. The
feasible set $\bigcap_{i\in\mathcal{N}}\mathcal{P}_{i}$ is closed
because each $\mathcal{P}_{i}$ is defined by a linear weak inequality.
It contains $q=0$, and $W(0)=0$. Hence its welfare supremum is unchanged
if we restrict attention to allocations satisfying $W(q)\geq0$. The
resulting set is a nonempty closed subset of the compact upper contour
set assumed in Section \ref{sec:general_tic}, and is therefore compact.
Since $W$ is continuous, it attains a maximum on this set. \end{proof}

\end{document}